\documentclass[final]{alt2027} % camera-ready template (de-anonymized)

\title[On Collusion-Free Teaching in Planar Graphs]{On Collusion-Free Teaching in Planar Graphs}

\usepackage{times}
\usepackage{custom}
\usepackage{todonotes}
\altauthor{%
 \Name{Liana Khazaliya} \Email{liana.khazaliya@aalto.fi}\\
 \addr Aalto University, Espoo, Finland\\
 \addr Helsinki Institute for Information Technology, Helsinki, Finland
 \AND
 \Name{Clemens Kuske} \Email{clemens.kuske@hpi.de}\\
 \addr Hasso Plattner Institute, Potsdam, Germany
}

\begin{document}

\maketitle

\begin{abstract}%
Non-clashing teaching minimizes teaching complexity among models satisfying the Goldman--Mathias collusion-avoidance criterion.
While the Vapnik--Chervonenkis dimension ($\operatorname{VCdim}$) measures the complexity of learning from samples, non-clashing teaching dimension ($\operatorname{NCTD}$) measures the number of examples required
when a teacher can select them.
Whether $\operatorname{NCTD}$ is always bounded above by $\operatorname{VCdim}$ remains a fundamental open question in learning theory.

In this paper, we investigate this question for the concept classes of closed neighborhoods in planar graphs.
We prove a tight upper bound of two on their $\operatorname{NCTD}$, thereby establishing the conjectured inequality for this family. 
For the positive variant $\operatorname{NCTD}^+$, in which teaching sets contain only positive examples, we obtain a tight upper bound of four. 
Finally, we show that deciding whether either $\operatorname{NCTD}^+$ or $\operatorname{NCTD}$ is at most one is \NP-complete, even for concept classes of closed neighborhoods in planar graphs.
In particular, despite the upper bound of two on $\operatorname{NCTD}$ for such classes, distinguishing between $\operatorname{NCTD}$ values of one and two remains \NP-hard.
\end{abstract}


\section{Introduction}
Machine teaching studies the problem of identifying a target concept through a carefully selected set of labeled examples.
Unlike traditional learning models, where examples are typically sampled from a distribution, machine teaching allows a teacher to choose examples specifically to facilitate learning. 
The objective is to minimize the number of examples required for the learner to identify
the target concept. 
A fundamental requirement is that the teacher and learner should not communicate information through unintended channels. 
The
most commonly accepted way to formalize this requirement was proposed by  \cite{GoldmanMathias1996} through a collusion-avoidance criterion.

Among a variety of formal models of teaching, for example, the
classical teaching dimension model~\cite{goldman1995complexity}, the optimal teacher model~\cite{balbach2008measuring}, or recursive teaching~\cite{zilles2011models}, we focus on
  \emph{non-clashing teaching} that has been introduced and studied in papers by  \cite{kirkpatrick2019optimal, fallat2023batch}, and subsequent papers by \cite{chalopin2024nonclashing, ganian2025computational, bhore2026nonclashing}.
  The non-clashing teaching model is known to be optimal among batch teaching models satisfying Goldman--Mathias criterion.
Its central requirement is that no two distinct concepts can be consistent with the union of their teaching sets.

Formally, let $\calC\subseteq 2^{\mathcal X}$ be a finite binary concept class over a domain $\mathcal X$.
A teaching map $T$ assigns to each concept $C\in\calC$ a set $T(C)\subseteq\mathcal X$ of examples, labeled according to their membership in $C$.
The map is \emph{non-clashing} if, for every pair of distinct concepts $C,C'\in\calC$, $$(T(C)\cup T(C'))\cap(C\triangle C')\neq\emptyset.$$
We call such a map a \emph{non-clashing teaching map} ($\nctm$).
Its size is $\max_{C\in\calC}|T(C)|$, and the \emph{non-clashing teaching dimension} $\nctd(\calC)$ is the minimum size of an $\nctm$ for $\calC$.
In the \emph{positive} variant, teaching sets contain only positive examples, that is, $T(C)\subseteq C$.
The corresponding dimension is denoted by $\nctd^+(\calC)$.

\paragraph{Non-clashing teaching and VC dimension.}
An important open question by~\cite{kirkpatrick2019optimal,fallat2023batch} is whether the non-clashing teaching dimension is always upper-bounded with the Vapnik--Chervonenkis dimension ($\operatorname{VCdim}$) by \cite{vapnik1971uniform}.
Recall that a set $S\subseteq\mathcal X$ is \emph{shattered} by $\calC$ if, for every subset $S'\subseteq S$, there exists a concept $C\in\calC$ such that $C\cap S=S'$; and $\operatorname{VCdim}(\calC)$ is the maximum size of a set shattered by $\calC$.

While $\operatorname{VCdim}$ measures the combinatorial complexity of learning from samples, $\nctd$ measures the number of examples required when a teacher can select them.
The inequality $\nctd(\calC)\leq\operatorname{VCdim}(\calC)$, conjectured in the works \cite{kirkpatrick2019optimal,fallat2023batch,chalopin2024nonclashing,bhore2026nonclashing} on non-clashing teaching, would therefore establish a fundamental connection between teaching and learning complexity.
Although the question remains open in general, investigating it for structured concept classes seems a natural way to approach it.  It is also motivated by the numerous works studying $\operatorname{VCdim}$ of closed neighborhoods in graphs \cite{ducoffe2020diameter, kranakis1997vc, chalopin2023sample, coudert2024practical}.
See~\cite{foucaud2025parameterized} for a list of references on other well-studied graph-based concept classes.

\paragraph{Teaching in graphs.}
Graph-based concept classes provide a natural setting for such an investigation.
On the one hand, they exhibit combinatorial and geometric structure that can be exploited
when constructing teaching maps, see, e.g., in \cite{ganian2025computational, bhore2026nonclashing}.
On the other hand, they are expressive enough to represent arbitrary finite binary concept classes.
Indeed, every finite binary concept class can be represented
as a selected family of closed neighborhoods in a split graph,
preserving its non-clashing teaching dimension.
Thus, studying such neighborhood families is as general as
studying arbitrary finite binary concept classes.

This observation motivates the study of neighborhoods
under structural restrictions on the underlying graph.
A related setting is the concept class of
graph balls, where $B_r(v)=\{u\in V(G):\dist(u,v)\leq r\}$, see, e.g. \cite{chepoi2007covering}.
\cite{chalopin2023sample, chalopin2024nonclashing, ganian2025computational} studied
non-clashing teaching for graph balls and established
the conjectured inequality for several graph classes,
including trees and interval graphs.
\cite{chalopin2024nonclashing, bhore2026nonclashing} also identified planar graphs as an interesting
direction, noting that their ball systems have $\operatorname{VCdim}$
at most four, while obtaining comparable bounds on $\nctd$
remained open.

Subsequent work by \cite{ganian2025computational} was focused on positive setting for
non-clashing teaching for graph balls.
More recently, \cite{bhore2026nonclashing}
studied the class of closed neighborhoods, obtaining
algorithmic and hardness results together with upper
bounds of five and seven for the general and positive
$\nctd$, respectively, in planar graphs.

\paragraph{Closed neighborhoods of planar graphs.}
In this paper, we study both the general and positive variants of non-clashing teaching for concept classes consisting of closed neighborhoods in planar graphs. The positive variant refers to teaching models that allow learning only from positive examples; see, e.g., \cite{denis2001learning, angluin1980inductive, schwab2000learning} and, for non-clashing teaching in particular, \cite{fallat2023batch,chalopin2024nonclashing,ganian2025computational}.

We refer to the corresponding
dimensions as \textsc{Neighborhoods-NCTD} (\nnctd) and
\textsc{Positive Neighborhoods-NCTD} (\nnctdplus), respectively.
The corresponding decision problems are defined as follows:

\defproblem{\textsc{Planar \nnctd}}
{A planar graph $G$, a family $\calN$ of closed neighborhoods in $G$,
and an integer $k$.}
{Is $\nnctdpar(\calN)\leq k$?}

\noindent
The problem \textsc{Planar \nnctdplus} is defined analogously for deciding on $\nnctdpar^+(\calN)$.

For two distinct neighborhoods $N[u],N[v]\in\calN$,
we say that a vertex $x\in T(N[u])\cup T(N[v])$
\emph{distinguishes} the pair if
$x\in N[u]\triangle N[v]$.
The pair is \emph{distinguished} if such a vertex exists.
If $N[u]=N[v]$, then $u$ and $v$ are called \emph{twins}
and represent the same concept. We can keep only one
vertex of each such pair, and so assume that a given graph $G$ is twin-free. Lastly, for a graph $G$, let $\calN_G=\{N[v]:v\in V(G)\}$.

\paragraph{Our results.}
We establish that, for every planar graph $G$, $\nnctdpar(\calN_G)\leq 2$, which is tight.
This improves the previously known upper bound of five, and establishes the conjectured inequality between $\nctd$ and $\operatorname{VCdim}$ for closed neighborhoods of planar graphs.
For \textsc{Planar} \nnctdplus, we prove the tight bound $\nnctdpar^+(\calN_G)\leq 4$,
improving the previous bound of seven. 
Finally, we prove that both \textsc{Planar} \nnctd and \textsc{Planar} \nnctdplus are \NP-complete even for $k=1$.
In particular, distinguishing between $\nctd$ one and two remains \NP-hard for closed neighborhoods of planar graphs.

\paragraph{Organization of the paper}
In \cref{subsec:upperpos}, we prove a tight upper bound of four
on $\nnctdpar^+(\calN_G)$ for planar graphs $G$.
In \cref{subsec:uppergen}, we establish a tight upper bound of two
for the general variant $\nnctdpar(\calN_G)$.
In \cref{sec:hardness}, we show that deciding whether either
dimension is at most one is \NP-complete, even for closed neighborhoods
of planar graphs.
Finally, we conclude in \cref{sec:conclusion}.

\section{Tight Upper Bounds}\label{sec:upperbounds}
In this section, we establish tight upper bounds on the non-clashing
teaching dimension of closed neighborhoods in planar graphs.
\subsection{Positive Non-Clashing Teaching Dimension}\label{subsec:upperpos}

The positive variant, being arguably more intuitive, is considered first. Our main result in this subsection is the following bound.
\begin{theorem}\label{thm:planar}
For every planar graph $G$, $\operatorname{N-NCTD}^{+}(\calN_G)\leq 4$.
\end{theorem}

The proof relies on an auxiliary bound of Lemma~\ref{lem:outerplanar} for outerplanar graphs, which we establish first for 2-connected graphs in Claim~\ref{cl:two-connected};
and then extend to the general case.

\begin{lemma}\label{lem:outerplanar}
For every outerplanar graph $G$, $\nnctdpar^+(\calN_G)\leq 2$.
\end{lemma}

\begin{claim}\label{cl:two-connected}
For every 2-connected outerplanar graph $G$,
$\nnctdpar^+(\calN_G)\leq 2$.
\end{claim}
\begin{proof}
We first consider graphs with at most four vertices.
For $|V(G)|\leq 3$, all closed neighborhoods coincide, so the claim follows immediately.
If $|V(G)|=4$, then $G$ is either a cycle or a cycle with a chord. In the former case, label the vertices cyclically and set $T(v_i)=\{v_i,v_{i+1}\}$.
In the latter case, label the vertices so that $v_1v_3$ is the chord, and set $
T(v_1)=T(v_3)=\{v_2,v_4\}$, $T(v_2)=\{v_2,v_3\}$, and $T(v_4)=\{v_4,v_1\}$.
It is straightforward to verify that in both cases the maps are non-clashing.

Now suppose that $|V(G)|=n\geq 5$.
Fix an outerplanar embedding and label the vertices $v_1, \ldots, v_n$ in cyclic order along the outer face, with indices taken modulo $n$.
Set $T(v_i)=\{v_{i-1},v_{i+1}\}$.
Clearly, each teaching set has size two and is contained in the corresponding closed neighborhood.

Suppose, for a contradiction, that two distinct closed neighborhoods $N[v_i]$ and $N[v_j]$ are not distinguished.
Then $\{v_{i-1},v_{i+1}\}\subseteq N[v_j]$, $\{v_{j-1},v_{j+1}\}\subseteq N[v_i]$.
Consider the two paths between~$v_i$ and~$v_j$ along the outer cycle. Since $n\geq 5$, at least one contains at least three edges. Without loss of generality, orient the cycle so that this path passes through $v_{j+1}$
and $v_{i-1}$, in that order. 
Thus, $v_i,v_j,v_{j+1},v_{i-1}$ are distinct.
The conditions above imply that $v_iv_{j+1},v_jv_{i-1}\in E(G)$.
Contracting the other $v_i$--$v_j$ path and the $v_{j+1}$--$v_{i-1}$ segment  gives us a $K_4$ minor, so contradicts outerplanarity.

Hence, $T$ is a positive non-clashing teaching map of size at most two.
\end{proof}

\begin{proof}[of Lemma~\ref{lem:outerplanar}]
We construct a positive non-clashing teaching map $T$ of size at most two.
Since teaching sets in different connected components are nonempty and positive, it suffices to consider each connected component separately.

First, for every cut-vertex $v$, we set $T(v)=\{u,w\}$, where $u$ and $w$ are neighbors of $v$ belonging to distinct components of $G-v$.
Next, for every non-cut-vertex belonging to a 2-connected component, we assign the teaching set given by Claim~\ref{cl:two-connected}.
Finally, for every remaining vertex $v$, we set $T(v)=\{v\}$.
Clearly, all teaching sets are positive and have size at most two.

We show that $T$ is non-clashing.
First, consider a cut-vertex $v$ with $T(v)=\{u,w\}$.
Since~$u$ and~$w$ belong to distinct components of $G-v$, no vertex other than $v$ is adjacent to both.
Thus,~$T(v)$ distinguishes $N[v]$ from every other closed neighborhood.

Now consider two distinct non-cut-vertices $u$ and $v$.
If they belong to the same 2-connected component, their neighborhoods are distinguished by Claim~\ref{cl:two-connected}.
Otherwise, they belong to different 2-connected components
or lie outside any such component. 
In either case, their neighborhoods intersect in at most one cut-vertex, while each teaching set contains an example outside this intersection.
Hence, at least one of their teaching sets distinguishes the pair.

Therefore, $T$ is a positive $\nctm$ of size at most two, and hence, for an outerplanar graph $G$, $\nnctdpar^+(\calN_G)\leq 2$.
\end{proof}

The bound of Lemma~\ref{lem:outerplanar} is tight.
Consider the outerplanar graph $G$ in Fig.~\ref{fig:outerplanar-tight}.

Since $N[v_1]=N[v_2]\cup\{v_6\}$, $N[v_4]=N[v_5]\cup\{v_3\}$,
any positive non-clashing teaching map of size one must satisfy
$T(v_1)=\{v_6\}$ and $T(v_4)=\{v_3\}$.
However, these sets do not distinguish a pair $\la N[v_1], N[v_4]\ra$.
Thus, $\nnctdpar^+(\calN_G)=2$.

\begin{figure}[ht]
    \centering
    \begin{minipage}{0.35\textwidth}\centering
    \scalebox{.75}{\begin{tikzpicture}[
    vertex/.style={circle, fill=black, inner sep=1.5pt},
    every edge/.style={draw, thick}]
    \node[vertex, label=left:$v_1$] (v1) at (0,1.6) {};
    \node[vertex, label=left:$v_2$] (v2) at (1.4,.8) {};
    \node[vertex, label=left:$v_3$] (v3) at (1.4,-.8) {};
    \node[vertex, label=left:$v_4$] (v4) at (0,-1.6) {};
    \node[vertex, label=left:$v_5$] (v5) at (-1.4,-.8) {};
    \node[vertex, label=left:$v_6$] (v6) at (-1.4,.8) {};

    \draw (v1)--(v2)--(v3)--(v4)--(v5)--(v6)--cycle;
    \draw (v1)--(v3) (v4)--(v6) (v1)--(v6);
\end{tikzpicture}}
    \caption{An outerplanar graph $G$ with $\nnctdpar^+(\calN_G)=2$.}
    \label{fig:outerplanar-tight}
    \end{minipage}
    \hfil
    \begin{minipage}{0.325\textwidth}\centering
    {\scalebox{.75}{\begin{tikzpicture}[
    vertex/.style={circle, fill=black, inner sep=1.5pt},
    every edge/.style={draw, thick}]
    \node[vertex, label=above left:$v_1$] (v1) at (-1.3,1.3) {};
    \node[vertex, label=above right:$v_2$] (v2) at (1.3,1.3) {};
    \node[vertex, label=below right:$v_3$] (v3) at (1.3,-1.3) {};
    \node[vertex, label=below left:$v_4$] (v4) at (-1.3,-1.3) {};
    \node[vertex, label=above:$v$] (v) at (0,0) {};

    \draw (v1)--(v2)--(v3)--(v4)--(v1);
    \draw (v)--(v1) (v)--(v2) (v)--(v3) (v)--(v4);
\end{tikzpicture}}}
    \caption{A planar graph $G$ with $\nnctdpar^+(\calN_G)=4$.}
    \label{fig:planar-tight}    \end{minipage}
\caption*{\textit{Examples} showing that the upper bounds in Lemma~\ref{lem:outerplanar} and Theorem~\ref{thm:planar} are tight.}
\end{figure}

Now we proceed with the proof of Theorem~\ref{thm:planar}.

\begin{proof}[Theorem~\ref{thm:planar}]
    We begin by constructing a preliminary teaching map, which we subsequently modify to satisfy the non-clashing condition.

\noindent
\textbf{Preliminary Teaching Map.}
For each vertex $v$ with $\deg(v)\leq 3$, set $T(v)=N[v]$.
For vertices of degree greater than three, we use the following construction. 
Since $G[N[v]\setminus\{v\}]$ is outerplanar, Lemma~\ref{lem:outerplanar} provides a positive non-clashing teaching map of size at most two for its closed neighborhoods.
We extend each teaching set by adding $v$, obtaining sets of size at most three.

For every vertex $v$ with $\deg(v)\leq 3$, the choice $T(v)=N[v]$ distinguishes its neighborhood from every other distinct neighborhood.
Moreover, when $\deg(v)>3$, the constructed teaching sets also distinguish every pair $\langle N[u],N[w]\rangle$ with $u\in N[v] \backslash \{v\}$ and $w\notin N[v]$, since $v\in T(u)$ but $v\notin N[w]$.
Thus, after constructing the preliminary teaching map $T$, it remains to distinguish pairs of neighborhoods centered within $N[v]$.
We address these pairs by modifying the preliminary teaching sets in four steps.

First, we consider vertices that are twins in $G[N[v]]$ but not in $G$. 
Next, we handle vertices $u\in N[v]\setminus\{v\}$ satisfying $N[v]\subsetneq N[u]$, which we call \emph{saturating vertices}.
We then modify the teaching sets of the  {remaining vertices in $N[v]\setminus\{v\}$} and, finally, the {teaching set of $v$}.
Throughout these modifications, we ensure that all previously distinguished pairs remain distinguished.

\smallskip\noindent
\textbf{Twin vertices in $G[N[v]\setminus \{v\}]$.}
Let $X\subseteq N[v]\setminus\{v\}$ be a set of pairwise twins in $G[N[v]\setminus\{v\}]$. 
Since such vertices form a clique, $|X|\leq 3$; 
otherwise, $X$ together with $v$ would contain a $K_5$, contradicting planarity. 
We modify the teaching sets according to the size of $X$.

\noindent
\textit{Case 1: $|X|=2$.}
Since the vertices of $X$ have distinct closed neighborhoods in $G$, their symmetric difference contains a vertex outside $N[v]$. 
We add this vertex to the teaching set of the vertex whose neighborhood has it, distinguishing the pair while increasing the teaching set size to at most four.

\noindent
\textit{Case 2: $|X|=3$.}
First, no vertex in $N[v]\setminus(X\cup\{v\})$ can be adjacent to a vertex of $X$. 
Otherwise, it would be adjacent to all three vertices of $X$, forming a $K_5$ together with $v$.

For each $u\in X$, set $T(u)=\{u,v\}$.
Since $G$ is twin-free, the closed neighborhoods of vertices in $X$ are pairwise distinct. 
For each pair $u,w\in X$, choose a vertex in $N[u]\triangle N[w]$ and add it to the teaching set of the vertex whose neighborhood contains it.
Each teaching set receives at most two additional vertices, so its size remains at most four.

These sets distinguish vertices of $X$ from vertices outside $N[v]$ by $v\in T(u)$, from vertices in $N[v]\setminus(X\cup\{v\})$ by $u\in T(u)$, and from one another by the additional vertices.

\smallskip\noindent
\textbf{Saturating vertices.}
We call a vertex $u\in N[v]\setminus\{v\}$ \emph{saturating} if $N[v]\subsetneq N[u]$. So, distinguishing~$N[u]$ from~$N[v]$ requires a vertex from $N[u]\setminus N[v]$ in $T(u)$.

If $\deg(v)>4$, there is at most one saturating neighbor.
Indeed, two such neighbors, together with $v$, would form one part of a $K_{3,3}$, with three other neighbors of $v$ forming the other, contradicting planarity.

If $u$ is the unique saturating neighbor, it has no twin in $G[N[v]\setminus\{v\}]$. 
Thus, its teaching set has size at most three, and we add a vertex from $N[u]\setminus N[v]$ to distinguish the pair.

It remains to consider $\deg(v)=4$ with two saturating neighbors $u,w$. 
These vertices are twins in $G[N[v]\setminus\{v\}]$, so the previous step has already added a vertex outside $N[v]$ to one of their teaching sets, say $T(u)$. 
Hence, $N[u]$ and $N[v]$ are already distinguished. 
We add a vertex from $N[w]\setminus N[v]$ to $T(w)$, keeping both teaching sets of size at most four.

\smallskip\noindent
\textbf{Remaining vertices of $N[v]\setminus \{v\}$.}   
We now aim to distinguish $N[v]$ from neighborhoods centered within non-saturating vertices of $N[v]\setminus \{v\}$.

\noindent
\textit{Case 1: $\deg(v)=4$.}
Let $N[v]=\{v,v_1,v_2,v_3,v_4\}$ and set $T(v)=N[v]\setminus\{v\}$.
This distinguishes~$N[v]$ from every non-saturating neighbor,
saturating neighbors are handled by their own teaching sets.

Suppose $w\in V(G)\setminus N[v]$ satisfies $N[v]\setminus\{v\}\subseteq N[w]$.
We redefine $T(w)=\{w,v_1,v_2,v_3\}$.
This set distinguishes $N[w]$ from every other neighborhood centered outside $N[v]$: a third vertex adjacent to $v_1,v_2,v_3$, together with $v$ and $w$, would yield a $K_{3,3}$, contradicting planarity.
Moreover, $w\in T(w)\setminus N[v]$ distinguishes~$N[w]$ from~$N[v]$, while
$v\in T(u)\setminus N[w]$ distinguishes $N[w]$
from~$N[u]$ for every $u\in N(v)$.

\noindent
\textit{Case 2: $\deg(v)> 4$.}
Here, we first show that the outerplanar graph
$H=G[N[v]\setminus\{v\}]$ contains two vertices
$v_1$ and $v_2$ of degree at most two with at most one
common neighbor in $H=G[N[v]\setminus \{v\}]$.

Since $H$ is outerplanar, it can be extended to a maximal outerplanar graph $H'$.
As $|V(H')|\geq 5$, we can choose two nonadjacent degree-two vertices $v_1,v_2$ of $H'$.
Since~$v_1$ and~$v_2$ have degree two, two common neighbors for them would force the outer boundary of $H'$ to be a four-cycle, contradicting $|V(H')|\geq5$.
Thus, $v_1$ and $v_2$ have at most one common neighbor in $H'$, and hence also in $H$.

So, if $v_1$ and $v_2$ have a common neighbor $w$ in $H$ such that $N[v]\nsubseteq N[w]$, choose $v_3\in N(v)\setminus N[w]$.
Otherwise, choose any $v_3\in N(v)\setminus\{v_1,v_2\}$.
Set $T(v)=\{v,v_1,v_2,v_3\}$.

The vertex $v$ distinguishes $N[v]$ from every neighborhood centered outside $N[v]$.
Moreover, any neighbor containing both $v_1$ and $v_2$ in its closed neighborhood must be their unique common neighbor $w$. 
By the choice of $v_3$, this neighbor either saturates $v$ or is distinguished from $v$ by $v_3$. 
Thus, $T(v)$ distinguishes $N[v]$ from every non-saturating neighbor, while saturating neighbors are handled by their own teaching sets.

\smallskip
After these modifications, every pair of distinct closed
neighborhoods is distinguished by at least one of their
teaching sets. Moreover, all teaching sets remain positive
and have size at most four. Thus, the resulting teaching map
is non-clashing, proving that $\nnctdpar^+(\calN_G)\leq 4$.
\end{proof}

The upper bound for $\nnctdpar^+$ on planar graphs is also tight. Consider the graph $G$ in
Fig.~\ref{fig:planar-tight}. Since $N[v]=V(G)$ and
$N[v_i]=V(G)\setminus\{v_{i+2}\}$ (indices modulo $4$),
any valid positive non-clashing teaching map must take
$v_{i+2}$ in $T(v)$ for every $i\in[4]$.
Thus, $|T(v)|\geq 4$, and the upper bound yields
$\nnctdpar^+(\calN_G)=4$.

\subsection{(General) Non-Clashing Teaching Dimension}\label{subsec:uppergen}

In this subsection, we prove a tight upper bound of two on the non-clashing teaching dimension of closed neighborhoods in planar graphs.

\begin{theorem}\label{thm:gen2}
For every planar graph $G$, $\nnctdpar(\calN_G)\leq 2$.
\end{theorem}

The proof relies on two auxiliary lemmas.
Lemma~\ref{lem:small} shows that every concept class of size at most five admits
a non-clashing teaching map of size at most one.
And Lemma~\ref{lem:deg5map} establishes that every planar graph of minimum
degree at least five contains a vertex whose closed neighborhood
induces a family admitting a non-clashing teaching map of size at most one. 
Together, these results yield Theorem~\ref{thm:gen2}.

\begin{lemma}\label{lem:small}
    For every concept class $\mathcal C$ with $|\mathcal C|\leq 5$, we have $\nctd(\calC)\leq 1$.
\end{lemma}

\begin{proof}
First, repeatedly remove any concept $C$ whose label on some example $x$
is unique among the remaining concepts, setting $T(C)=\{x\}$.
This distinguishes $C$ from all remaining concepts.
If at most one concept remains, we are done.

Otherwise, the remaining family $\mathcal R$ has four or five concepts. 
Since the reduction can no longer be applied, every nonconstant example, i.e., an example whose label is not the same for all concepts in $\mathcal R$, must be positive for at least two concepts and negative for at least two others. 
Thus, each such example partitions $\mathcal R$ into two label classes of sizes $2+2$ or $2+3$.

Construct an auxiliary graph $H$ on $\mathcal R$, where each nonconstant example $x$ induces an edge $e_x$ joining the two concepts from one of its size-two label classes. 
Thus, $x$ distinguishes $C_1,C_2\in\mathcal R$ precisely when exactly one of them is an endpoint of $e_x$.

Since the concepts are distinct, $H$ has at most one isolated vertex and no component of size two, as either would imply two concepts with identical labels on all examples. 
Hence, $H$ has one nontrivial component, possibly with an isolated vertex $I$.

Let $T$ be a spanning tree of the nontrivial component, rooted at a leaf $R$. Let $S$ be the neighbor of $R$, and choose a neighbor $K\neq R$ of $S$. For each non-root vertex $C$, assign the example corresponding to its parent edge in $T$.
Assign to $R$ the example corresponding to an edge $SK$, and, if $I$ exists, assign to $I$ the example corresponding to an edge $RS$.

We verify that the resulting map is non-clashing.
Two non-root vertices are distinguished by at least one of their parent edges, since they cannot be parents of each other (every edge has a unique example, since $H$ was constructed on size-two sets). 
The vertex $R$ is distinguished from $S$ by $SK$, and from every other non-root vertex by its parent edge.
Finally, if $I$ exists, it is distinguished from every non-root vertex by that vertex's parent edge, and from $R$ by an edge $RS$.

Thus $\mathcal R$ admits a non-clashing teaching map of size one.
Inserting the previously removed concepts preserves the non-clashing property, completing the proof.
\end{proof}

Before proceeding with Lemma~\ref{lem:deg5map}, we need to prove the following structural property of planar graphs with minimum degree~$5$.
\begin{lemma}\label{lem:structural}
    Let $G$ be a planar graph with minimum degree~$5$.
    Then there exists a vertex $v\in V(G)$ such that
    $\deg(v)=5$ and
    $G[N[v]\setminus \{v\}]\in\{P_5, C_5\}$.
\end{lemma}
\begin{proof}
    We may assume that $G$ is connected. 
    Every planar graph $G$ with minimum degree at least $5$ contains a triangular face; otherwise, Euler's formula implies $|E(G)| \leq 2|V(G)| - 4$, contradicting the Handshaking Lemma, which gives $|E(G)| \geq \frac{5}{2}|V(G)|$.
    Fix a plane embedding of $G$ in which a triangular face $T$ is the outer face. Among all triangles whose bounded interiors contain at least one vertex, choose a triangle $T$ minimizing the number of vertices in its interior.
    By the minimality of $T$, there is no non-empty triangle contained in the face bounded by $T$.
    
    Let $H$ be the subgraph consisting of $T$, together with all vertices and edges contained in the face bounded by $T$. 
    We denote $I=V(H)\setminus V(T)$ (vertices enclosed by $T$), and let $F$ denote the set of bounded faces of $H$.
    Furthermore, for every $u\in I$, $\deg_H(u)=\deg_G(u)\geq 5.$
    For convenience, we also denote $|I|=n$, $|E(H)|=m$, $|F|=f$ and for a face $f' \in F$ let $\ell(f')$ be the length of the boundary walk.
    First, we make two standard observations for $H$.
By Euler's formula, we have $(n+3)-m+(f+1)=2$,
so $f=m-n-2$.
Also, $\sum\limits_{f\in F}\ell(f)=2m-3$,
since each edge is counted twice, except for the three edges of the outer face.

Define the total degree deficit of interior vertices as $D=\sum\limits_{u\in I}(6-\deg(u))$.
Since every interior vertex has degree at least~$5$,
only degree-$5$ vertices contribute positively to $D$. 
We also define the faces sum as $Q=\sum\limits_{f\in F}(2\ell(f)-6)$, where only triangular faces contribute zero.

By the Handshaking Lemma,
\begin{equation}\label{eqD}
    D=\sum\limits_{u\in I}(6-\deg(u))=6n-\sum\limits_{u\in I}\deg(u)=6n-2m+\sum\limits_{v\in T}\deg_H(v).
\end{equation}

Based on the observations above, 
\begin{equation}\label{eqQ}
    Q=\sum\limits_{f\in F}(2\ell(f)-6)=2\sum\limits_{f\in F}\ell(f)-6|F|=2(2m-3)-6(m-n-2)=6n-2m+6.
\end{equation}

Therefore, from Equations~\ref{eqD} and~\ref{eqQ}, $D=Q-6+\sum\limits_{v\in T}\deg_H(v)$.
Since the sum $\sum\limits_{v\in T}\deg(v)$ counts edges of $T$ twice and connectivity guarantees at least one more edge
between $T$ and $I$, we have
$\sum\limits_{v\in V(T)}\deg_H(v)>6$.
Hence, $D>Q$.

Assume that there is no degree-$5$ vertex in $I$ incident with at most one non-triangular face. 
In other words, for every vertex $u\in I$ of degree~$5$, the number of its incident non-triangular faces~$t(u)\geq 2$.

First, let us upper-bound the sum of these terms:
\begin{equation}\label{eq1}
    \sum\limits_{u\in I}t(u)\leq \sum\limits_{\substack{f\in F\\\ell(f)\geq 4}}\ell(f). 
\end{equation}
Indeed, on the left-hand side, each face $f$ is counted only if $\ell(f)>3$, and at most $\ell(f)$ times.

Next, under our assumption, the inequality
$6-\deg(u)\leq \frac{t(u)}{2}$ holds for every $u\in I$.
Therefore, $D=\sum\limits_{u\in I}(6-\deg(u))\leq \frac{1}{2}\sum\limits_{u\in I}t(u)$, and by the upper-bound we established above, we further obtain
$\frac{1}{2}\sum\limits_{u\in I}t(u)\stackrel{\textcolor{red}{\ref{eq1}}}{\leq} \frac{1}{2}\sum\limits_{\substack{f\in F\\\ell(f)\geq 4}}\ell(f)$.

Lastly, let us show that $\frac{1}{2}\sum\limits_{\substack{f\in F\\\ell(f)\geq 4}}\ell(f)\leq \sum\limits_{f\in F}(2\ell(f)-6) = Q$.
For every triangular face $f$, we have
$2\ell(f)-6=0$, while for every face $f$
with $\ell(f)\geq4$, we have
$\frac12\ell(f)\leq 2\ell(f)-6$.

Thus, under our assumption, $D\leq Q$, contradicting $D>Q$.
Therefore, there  always exists a vertex $v\in V(G)$ of degree~$5$ incident with at least four triangular faces.

\medskip
Consider the graph $H'=G[N[v]\setminus \{v\}]$.
    Since at least four faces are triangular, there is a path through all five vertices of $H'$,     which can be ordered along this path.
   Assume that there is an edge between two vertices $u$, $w$ that are non-consecutive in the cyclic order     around $v$. 
   Then the edges between $u,w$, and $v$ would form     a separating triangle, yielding a contradiction with the minimality of chosen triangular outer face $T$.
   Thus, the only possible additional edge joins     the endpoints of the path.
    Therefore, $H'$ is either a path or a cycle     on five vertices.
\end{proof}

\begin{lemma}\label{lem:deg5map}
    For every planar graph $G$ with minimum degree $5$,
there exists a vertex $v\in V(G)$ such that the family
$\calN_v=\{N[u]\mid u\in N[v]\}$ admits a non-clashing
teaching map of size at most $1$.
\end{lemma}

\begin{proof}
    Let $v\in V(G)$ be a vertex of degree five such that $G[N(v)]\in\{ P_5, C_5\}$ (exists by Lemma~\ref{lem:structural}).

First, suppose $G[N(v)]\cong C_5$, with vertices indexed cyclically modulo five, i.e. $\{v_1,\ldots, v_5\}$. Define $T(N[v])=\emptyset$, $T(N[v_i])=\{v_{i+2}\}$, for each $v_i\in N[v]\setminus \{v\}$.
Every pair involving $N[v]$ is distinguished, since
$v_{i+2}\in N[v]\setminus N[v_i]$.
For two consecutive vertices, $\la N[v_i], N[v_{i+1}]\ra$ is distinguished by
$v_{i+2}$.
For vertices at cyclic distance two, one selects the other, which also distinguishes the pair.

Now suppose $G[N(v)]\cong P_5$, with vertices
$v_1,\ldots,v_5$ in path order. Define
\begin{align*}
T(N[v])   &= \{v_1\}, &
T(N[v_1]) &= \{v_3\}, &
T(N[v_2]) &= \{v_4\},\\
T(N[v_3]) &= \{v_5\}, &
T(N[v_4]) &= \{v_3\}, &
T(N[v_5]) &= \{v_2\}.
\end{align*}
Every pair involving $N[v]$ is distinguished by the teaching
set of the neighbor, except for $N[v_4]$, which is
distinguished by $v_1\in T(N[v])$.

Among the neighbors, $T(N[v_1])$ distinguishes $N[v_1]$
from $N[v_2]$, $N[v_3]$, $N[v_4]$;
$T(N[v_2])$ distinguishes $N[v_2]$ from
$N[v_3]$, $N[v_4]$, $N[v_5]$;
and $T(N[v_3])$ distinguishes $N[v_3]$ from
$N[v_4]$, $N[v_5]$.
The two remaining pairs are distinguished by
$T(N[v_5])$ and $T(N[v_4])$, respectively.

In both cases, all teaching sets belong to $N[v]$,
so the assignment is valid in $G$.
Hence the constructed maps are non-clashing and have
size at most one.
\end{proof}

Now, we are turning to the proof of Theorem~\ref{thm:gen2}.

\begin{proof}[Theorem~\ref{thm:gen2}]
    First, suppose that $G$ contains a vertex $u$ with $\deg(u)\leq 4$.
By Lemma~\ref{lem:small}, the family $\calN_u=\{N[v]\mid v\in N[u]\}$ admits a non-clashing teaching map $S$ of size at most one.
For every $v\in N[u]$, set $T(N[v])=S(N[v])\cup\{u\}$.
The example $u$ distinguishes every pair $\langle N[v],N[w]\rangle$ with $v\in N[u]$ and $w\notin N[u]$, while $S$ distinguishes all pairs of distinct neighborhoods centered in~$N[u]$.

Otherwise, $G$ has minimum degree at least five.
By Lemma~\ref{lem:deg5map}, there exists a vertex $u$ of degree five such that $G[N(u)]\cong P_5$ or $C_5$, and the family $\calN_u=\{N[v]\mid v\in N[u]\}$ admits a non-clashing teaching map $S$ of size at most one.
For every $v\in N[u]$, set $T(N[v])=S(N[v])\cup\{u\}$.
As in the previous case, the example $u$ distinguishes every pair $\langle N[v],N[w]\rangle$ with $v\in N[u]$ and $w\notin N[u]$, since $u\in N[v]\setminus N[w]$.
Meanwhile, $S$ distinguishes all pairs of distinct neighborhoods centered in~$N[u]$, including those involving $N[u]$ itself.
Thus, all these pairs are distinguished using teaching sets of size at most two.

In either case, remove $N[u]$ from $G$ and repeat the construction on the remaining induced graph.
Since each iteration distinguishes all pairs within the removed set and between the removed set and the remaining vertices, the resulting teaching map is non-clashing and has size at most two.
\end{proof}

The bound of Theorem~\ref{thm:gen2} is tight, as shown by the example in Subsection~\ref{subsec:upperpos}
(see Figure~\ref{fig:outerplanar-tight}).

\section{Hardness of Teaching Neighborhoods}\label{sec:hardness}
This section establishes that deciding whether the non-clashing teaching dimension is at most one is \NP-complete for closed neighborhoods of planar graphs, in both the general and positive settings.

\begin{theorem}\label{thm:hardplanar}
    Given a planar graph $G$ and a set of closed neighborhoods $\calN$, deciding whether $\nnctdpar(\calN)\leq 1$ is \NP-complete.
\end{theorem}

\begin{figure}[ht]
    \centering
    \begin{minipage}{0.35\textwidth}
    {\scalebox{.9}{\begin{tikzpicture}[
    vertex/.style={circle, fill=black, inner sep=1.5pt},
    every edge/.style={draw, thick}
]

%  \draw[color=gray!30, fill=gray!30] (0.2,1.3) ellipse (2.3cm and 1.6cm);
%\draw[color=gray!60, fill=gray!60] (0,-.05) ellipse (1.5 cm and .5cm);

% Vertices
\node[vertex, label=left:$t_1$] (t1) at (3,10.25) {};
\node[vertex, label=left:$x_1$] (x1) at (3,9.5) {};
\node[vertex, label=left:$f_1$] (f1) at (3,8.75) {};

\node[vertex, label=left:$t_2$] (t2) at (3,8) {};
\node[vertex, label=left:$x_2$] (x2) at (3,7.25) {};
\node[vertex, label=left:$f_2$] (f2) at (3,6.5) {};

\node[vertex, label=left:$t_3$] (t3) at (3,5.5) {};
\node[vertex, label=left:$x_3$] (x3) at (3,4.75) {};
\node[vertex, label=left:$f_3$] (f3) at (3,4) {};

\node[vertex, label=left:$t_4$] (t4) at (3,3.25) {};
\node[vertex, label=left:$x_4$] (x4) at (3,2.5) {};
\node[vertex, label=left:$f_4$] (f4) at (3,1.75) {};

\node[vertex, label=below:$c_2$] (c2) at (0.5,5) {};
\node[vertex, label=above:$c_1$] (c1) at (5.5,8) {};

% Internal edges
\foreach \v in {1,2,3,4} {
    \draw (t\v) -- (x\v) -- (f\v);
}
\draw (c1) -- (t1)
(c1) -- (f2)
(c1) -- (t3);
\draw (c2) -- (f3)
(c2) -- (f2)
(c2) -- (t4);

% Outer edges

\end{tikzpicture}}}
    \caption{Associated graph $G_{\varphi}$ for a \planarsat instance $\varphi=(x_1\vee \overline{x_2}\vee x_3)\wedge (\overline{x_2}\vee \overline{x_3}\vee x_4)$.}\label{fig:add-on}
    \end{minipage}
    \hspace{2.3cm}
    \begin{minipage}{0.425\textwidth}
    {\scalebox{.9}{\begin{tikzpicture}[
    vertex/.style={circle, fill=black, inner sep=1.5pt},
    every edge/.style={draw, thick},
    vertw/.style={circle, draw=black, fill=white, inner sep=1.5pt}
]

%\draw[color=gray!30, fill=gray!30] (0.2,1.3) ellipse (2.3cm and 1.6cm);
%\draw[color=gray!60, fill=gray!30] (0,5) ellipse (.6 cm and .32cm);
%\node[label=left:{\scriptsize$F(c_2)$}] (fc2) at (.5,5) {};
%\draw[color=gray!60, fill=gray!30] (6.385,8.15) ellipse (1 cm and 1.3cm);
%\node[label=right:{\scriptsize$F(c_1)$}] (fc1) at (6.25,8.7) {};

% Vertices
%\draw[color=gray!60, fill=gray!30] (3.5,9.5) ellipse (.6 cm and .32cm);
%\node[label=right:{\scriptsize$F(x_1)$}] (fx1) at (3,9.5) {};
\node[vertw, label=left:$f_1$] (t1) at (3,10.25) {};
\node[vertex, label=left:$x_1$] (x1) at (3,9.5) {};
\node[vertex, label=below:{\scriptsize$x'_1$}] (x'1) at (3.5,9.25) {};
\node[vertw, label=left:$t_1$] (f1) at (3,8.75) {};

%\draw[color=gray!60, fill=gray!30] (3.5,7.25) ellipse (.6 cm and .32cm);
%\node[label=right:{\scriptsize$F(x_2)$}] (fx2) at (3,7.25) {};
\node[vertw, label=right:$f_2$] (t2) at (3,8) {};
\node[vertex, label=right:$x_2$] (x2) at (3,7.25) {};
\node[vertex, label=above:{\scriptsize$x'_2$}] (x'2) at (2.5,7.5) {};
\node[vertw, label=right:$t_2$] (f2) at (3,6.25) {};
\node[vertw, label=right:$t_2'$] (t2') at (2.2,6.35) {};
%\draw[color=gray!60, fill=gray!30] (1.45,6.02) ellipse (.35 cm and .47cm);
\node[vertex, label=left:$t_2^2$] (t22) at (1.3,6.3) {};
\node[label=above right:{\scriptsize${t'}_2^2$}] (ftd22) at (1.25,5.5) {};
\node[vertex] (ft'22) at (1.3,5.65) {};

%\draw[color=gray!60, fill=gray!30] (3.5,4.75) ellipse (.6 cm and .32cm);
%\node[label=right:{\scriptsize$F(x_3)$}] (fx4) at (3,4.75) {};
\node[vertex, label=below:{\scriptsize$x'_3$}] (x'3) at (3.5,4.5) {};
\node[vertw, label=left:$f_3$] (t3) at (3,5.5) {};
\node[vertex, label=left:$x_3$] (x3) at (3,4.75) {};
\node[vertw, label=left:$t_3$] (f3) at (3,4) {};

%\draw[color=gray!60, fill=gray!30] (3.5,2.5) ellipse (.6 cm and .32cm);
\node[vertex, label=below:{\scriptsize$x'_4$}] (x'4) at (3.5,2.25) {};
\node[vertw, label=left:$f_4$] (t4) at (3,3.25) {};
\node[vertex, label=left:$x_4$] (x4) at (3,2.5) {};
\node[vertw, label=left:$t_4$] (f4) at (3,1.75) {};

\node[vertex, label=left:$c_2$] (c2) at (0.5,5) {};
\node[vertex, label=below:$c'_2$] (c'2) at (0.75,4.25) {};
\node[vertex, label=right:$c_1$] (c1) at (5.5,8) {};
\node[vertex, label=below:$c'_1$] (c'1) at (5.25,7.25) {};

\node[vertex, label=above:{$v_0$}] (a) at (4.5,4) {};
\node[vertex, label=below:{$v_1$}] (e) at (5.9,4) {};
\node[vertex, label=right:{$v_2$}] (c11) at (5.2,4.8) {};
\node[vertex, label=left:{$v_3$}] (c14) at (5.2,3.2) {};

% Internal edges
\foreach \v in {1,2,3,4} {
    \draw (t\v) -- (x\v) -- (f\v);
}
\draw (c'1) -- (c1) -- (t1)
(c1) -- (f2)
(c1) -- (t3);
\draw (x1) -- (x'1)
(x2) -- (x'2)
(x3) -- (x'3)
(x4) -- (x'4)
;
\draw (a) -- (c11) -- (e)
(a) -- (c14) -- (e)
(t22) -- (ft'22);
\draw (c'2) -- (c2) -- (f3)
(c2) -- (f2)
(c2) -- (t4);
\draw (c2) -- (t22) -- (x2) (t22) -- (t2') -- (x2);
\draw (c2) .. controls (.5,7.5) and (1.6,6.5) .. (x2);

% Outer edges

\end{tikzpicture}}}
    \caption{A graph $G$ constructed from the associated graph $G_{\varphi}$. Black vertices represent concepts included in $\calN$.} \label{fig:reduction}    \end{minipage}
    \caption*{\textit{Reduction Example}: from an instance $\varphi=(x_1\vee \overline{x_2}\vee x_3)\wedge (\overline{x_2}\vee \overline{x_3}\vee x_4)$ of \planarsat to an instance $(G, \calN, 1)$ of \textsc{Planar} \nnctd.}
\end{figure}

We reduce from \planarsat, which is one of the known \NP-complete restrictions of \textsc{Planar 3-SAT}~\cite{Dahlhaus94, Tippenhauer2016}.

\defproblem{\planarsat}{A CNF formula $\varphi=(\calX,\calC)$ such that
(i) every clause contains two or three literals over distinct variables,
(ii) every variable occurs exactly once positively and twice negatively, and
(iii) the variable--clause incidence graph is planar.}{Is $\varphi$
satisfiable?}

\paragraph{The reduction} 
Given an instance $\varphi=(\calX, \calC)$ of \planarsat, we construct \textsc{Planar} \nnctd instance~$(G, \calN, 1)$ in polynomial time, starting with $\varphi=(\calX,\calC)$ as an instance of \planarsat, for an example, see Fig.~\ref{fig:add-on}. 
The associated graph $G_{\varphi}$ has a variable vertex~$x_i$,
two literal vertices $t_i$ and $f_i$ for each variable, and a clause vertex $c_j$ for each clause.
Each~$x_i$ is adjacent to $t_i$ and $f_i$, and each literal vertex is adjacent to the clause vertices representing its occurrences.

First, attach one pendant vertex $x'_i$ to each variable vertex $x_i$ and
one pendant vertex $c'_j$ to each clause vertex $c_j$.
Next, for every pair of a true-literal vertex $t_i$ and a false-literal vertex $f_i$, we reverse their clause--literal adjacencies in $G_{\varphi}$. 
This can be done by simply exchanging the labels of these two vertices. 
After this modification, each vertex $t_i$ is adjacent to the two clauses containing
$\overline{x_i}$, whereas $f_i$ is adjacent to the unique clause containing
$x_i$.

For each variable $x_i$, choose one of its two negative-occurrence clauses,
say $c_j$. Add the edge~$c_jx_i$ and a new vertex $t_i^j$ adjacent to
both $c_j$ and $x_i$. We call $t_i^j$ the \emph{bonus clause-vertex}.
Attach a pendant vertex ${t'}_i^j$ to $t_i^j$, and add an \emph{extra vertex}
$t'_i$ adjacent to both $x_i$ and $t_i^j$. See Figure~\ref{fig:reduction} for an illustration.

These changes preserve planarity. Indeed, $t_i$ has exactly three neighbors
in $G_{\varphi}$. Thus, the edge~$c_jx_i$ can be drawn alongside the path $c_jt_ix_i$ within a face of the given embedding. The vertices $t_i^j$, ${t'}_i^j$, and $t'_i$, together with their incident edges, can be embedded
inside the resulting triangular face.
Pendant vertices can be embedded locally.
The remaining elements can then be embedded without crossings inside the face bounded by this path and the newly added edge~$c_jx_i$, since none of the newly introduced vertices is adjacent to any vertex outside this face (see Fig.~\ref{fig:reduction} for an illustration).
Lastly, we add a disjoint 4-cycle $C$, $V(C)=\{v_0, v_1, v_2, v_3\}.$
This completes the construction of $G$.

Define $\calN$ to consist of the closed neighborhoods of all vertices of $G$, except those of the literal-vertices $t_i,f_i$ and the extra vertices $t'_i$. Thus $\calN$ includes neighborhoods centered at variable vertices, clause vertices, bonus
clause-vertices, pendant vertices, and the 4-cycle vertices. 
Finally, we return~$(G,\mathcal{N},1)$ as an instance of \textsc{Planar} \nnctd and the construction is polynomial in the size of $\varphi$.

The high-level idea behind reversal of the clause--literal adjacencies is the following.
Consider a clause-vertex $c_j$ and a variable-vertex $x_i$ such that the variable $x_i$ appears in the clause corresponding to $c_j$.
If the teaching set of $N[x_i]$ contains a literal-vertex representing an assignment that satisfies this clause, then this literal-vertex should distinguish $N[x_i]$ from $N[c_j]$. 
To ensure this, we reverse the clause--literal adjacencies so that the satisfying literal-vertex is not adjacent to the clause-vertex $c_j$.

\begin{proof}
Let us first show the following direction: if $(G,\calN,1)$ admits a valid non-clashing teaching map, then the \planarsat instance $\varphi$ is satisfiable.

Suppose that $\calN$ admits a non-clashing teaching map $T$ of size at most one.
First, consider neighborhoods of the vertices of the 4-cycle $C$. For indices taken modulo $4$, there exists $v_i\in V(C)$ such that, in order to distinguish $N[v_{i}]$ from $N[v_{i+1}]$ and from $N[v_{i-1}]$, $T(v_i)=\{v_{i+2}\}$, since no other vertex distinguishes both pairs.
In particular, $v_{i+2}$ belongs neither to $N[v_i]$ nor to any
neighborhood centered outside $C$. Consequently, for any
$u\notin V(C)$, $T(u)$ must be a subset of~$N[u]$, i.e. if $T(u)=\{z\}$ with
$z\notin N[u]\cup V(C)$, then $N[u]$ and $N[v_i]$ would not be distinguished, and if $z\in V(C)$, $N[u]$ and $N[v_j]$ are not distinguished, where $v_j$ is a non-adjacent to $z$ vertex of the cycle.

For each clause-vertex $c_j$, a pair $\la N[c_j], N[c'_j] \ra$ can be distinguished only by $T(c_j)\subseteq N[c_j]\setminus \{c_j, c'_j\}$. The same applies for each variable-vertex and each bonus clause-vertex. 

Then, for each variable-vertex $x_i$, there is a vertex $t_i^j$.
Since  $N[t_i^j]\subset N[x_i]$ and $N[x_i]\setminus N[t_i^j]=\{t_i, f_i\}$,
we must have $T(x_i)\cap\{t_i, f_i\}\neq \emptyset$.
Consequently, $T(x_i)$ is either $\{t_i\}$ or $\{f_i\}$.

We now define a truth assignment for $\varphi$ by setting
$x_i=\true$ if $T(x_i)=\{t_i\}$, and
$x_i=\false$ if $T(x_i)=\{f_i\}$.

Consider a clause-vertex $c_j$.
As we observed, $T(c_j)\subseteq N[c_j]\setminus \{c_j, c'_j\}$.
Every vertex in this set belongs to
$N[x_i]$ for at least one variable $x_i$ occurring in $c_j$, i.e.  such a vertex is either a literal vertex, the relevant variable vertex, or a bonus clause-vertex.
So, let us consider such $x_i$.
Since $T(c_j)\subseteq N[x_i]$, the pair $N[c_j],N[x_i]$ must be distinguished by $T(x_i)$.
So, the literal vertex chosen to $T(x_i)$ is therefore absent from $N[c_j]$, which, by the
reversed clause--literal adjacencies, means that the occurrence of $x_i$ in $c_j$ is satisfied. This holds for every clause. 
Hence, $\varphi$ is
satisfiable.

\medskip\noindent
For the reverse direction, suppose that $\varphi$
admits a satisfying assignment.
We define a teaching map~$T$ of size~$1$ for
$(G,\calN,1)$ as follows.

For every variable-vertex $x_i$, set $T(x_i)=
    \begin{cases}
        \{t_i\}, & \text{if }x_i=\true,\\
        \{f_i\}, & \text{if }x_i=\false.
    \end{cases}$
    
For every pendant vertex $u$, $T(u)=\{u\}$.
For every bonus clause vertex, set $T(t_i^j)=\{t'_i\}$.
For every vertex $v_i\in V(C)$, set
$T(v_i)=\{v_{i+2}\}$, with indices modulo $4$.

It remains to assign teaching examples to the clause-vertices.
Consider a clause-vertex $c_j$.
Since the assignment satisfies $\varphi$, there exists
a variable $x_i$ whose assignment satisfies the clause
represented by $c_j$.
If several such variables exist, choose one arbitrarily.
So, if $x_i$ satisfies $c_j$ and is assigned $\true$, by reversed adjacencies, $f_i\in N[c_j]$, and we set $T(c_j)=\{f_i\}$. 

Recall that each variable appears exactly once positively
and exactly twice negatively in $\varphi$.
Hence, for each variable $x_i$, there can be at most one
pair of clauses $c_j,c_{\ell}$ that are both satisfied
by the same assignment to $x_i$.
Such a pair can arise only when $x_i=\false$.
So, if $c_j$ (and then $c_{\ell}$) is satisfied by $x_i=\false$, we set $T(c_{\ell})=\{t_i\}$ and we set $T(c_j)=\{t^j_i\}$ (assuming that among~$c_{\ell}$ and~$c_j$ the unique for each such a pair bonus clause vertex is adjacent to $c_j$). 

We claim that $T$ is a valid non-clashing teaching map
for $(G,\calN,1)$.

Any neighborhood centered in $V(C)$ is distinguished from any neighborhood centered in $V(G)\setminus V(C)$ by the positive teaching example of the latter.
Any two neighborhoods centered in~$V(C)$ are distinguished because each cycle vertex selects its opposite vertex.
The teaching set of a pendant concept belongs only to the
neighborhoods centered at the vertex itself and its parent; the parent is
distinguished from the pendant by teaching set of the parent including neither of them.
Thus all pairs involving a pendant concept are distinguished.

Consider two distinct variable vertices $x_i$ and $x_r$.
By construction, no $t_i$ or $f_i$ belongs to $N[x_r],$
so all such closed neighborhoods are distinguished by their teaching sets. By the same argument, neighborhoods centered in variable-vertices are distinguished from those centered in bonus clause vertices.

Next, consider bonus clause vertices. Neighborhoods of two distinct bonus clause vertices are distinguished, since one is never adjacent to a vertex in the teaching set of the other. The same holds for pairs across bonus clause vertices and clause-vertices, i.e. $t'_i\in T(t_i^j)$, but $t'_i\notin N[c_\ell]$ for each clause-vertex $c_{\ell}$.

Consider two distinct clause-vertices $c_j$ and
$c_{\ell}$.
Their neighborhoods are distinguished because their
teaching sets either belong to the non-intersecting neighborhoods of
distinct variable-vertices, or one of the teaching
sets contains a vertex $t_i^j$ that is not adjacent to
any other clause vertex.

Finally, consider a clause vertex $c_j$ and a variable
vertex $x_p$.
The example in $T(c_j)$ is associated with some variable $x_i$ occurring in $c_j$ and belongs to $N[x_i]$. If $p\ne i$, this example does not
belong to $N[x_p]$, so the pair is distinguished. If $p=i$, then the
example chosen for $c_j$ came from an occurrence of a variable that received an assignment satisfying the clause $c_j$.
In this case, the opposite literal-vertex was selected by $T(x_i)$. Therefore, $T(x_i)\not\subseteq N[c_j]$, and the pair is again distinguished.

Thus, every pair of distinct neighborhoods in $\calN$
is distinguished by their teaching sets, so $T$ is a non-clashing teaching map of size one for $(G,\calN,1)$, completing the
proof.
\end{proof}

A similar result can also be stated for the positive non-clashing teaching dimension.
\begin{theorem}\label{thm:hardplanarplus}
    Given a planar graph $G$ and a set of closed neighborhoods $\calN$, deciding whether $\nnctdpar^+(\calN)=1$ is \NP-complete.
\end{theorem}

Indeed, if $u'$ is a pendant neighbor of $u$, then distinguishing $N[u]$ from $N[u']$ in a positive non-clashing teaching map of size at most one requires $T(u)\subseteq N(u)\setminus\{u'\}$.
Thus, the 4-cycle can be omitted, and the forward implication follows as in the proof of Theorem~\ref{thm:hardplanar}.
Conversely, the teaching map constructed for the remaining neighborhoods in $\calN$ is already positive and non-clashing.

\section{Conclusion}
\label{sec:conclusion}

We proved tight upper bounds for both the general and positive
non-clashing teaching dimensions of closed neighborhoods in
planar graphs. We also showed that deciding whether either
dimension is at most one remains \NP-complete, even for planar
graphs.
Our results prove the inequality between $\nctd$ and $\operatorname{VCdim}$ that was conjectured for this family by \cite{kirkpatrick2019optimal}, and highlighted as an open question later by~\cite{chalopin2024nonclashing}.

However, whether
$\nctd(\calC)\leq\operatorname{VCdim}(\calC)$
holds for every finite concept class $\calC$ remains open.
A natural direction is to extend our results to graph
classes of linear neighborhood complexity, see, e.g., \cite{reidl2019characterising}. Such classes
have bounded $\operatorname{VCdim}$, making them natural candidates
for studying this relationship beyond planar graphs.

\cite{moran2016sample, doliwa2014recursive} studied
relationships between recursive teaching dimension,
$\operatorname{VCdim}$, and sample compression.
It would be interesting to determine whether similar
relationships hold for non-clashing teaching and whether
bounds on $\nctd$ could lead to improved sample compression schemes.

\acks{The authors would like to thank Fionn {Mc Inerney} for insightful comments and valuable feedback. This work has been supported by Helsinki Institute for Information Technology HIIT.}

% Replace yourbibfile with the name of your BibTeX database and uncomment:
 \bibliography{bib}

%\appendix

% \crefalias{section}{appendix} % uncomment if you are using cleveref

%\section{Deferred proofs}

\end{document}